\documentclass[conference]{IEEEtran}
\usepackage[linesnumbered, ruled]{algorithm2e}
\usepackage{amsmath, amsfonts, amsthm}
\usepackage[hidelinks]{hyperref}
\usepackage[utf8]{inputenc}

\newtheorem{thm}{Theorem}
\newtheorem{lem}[thm]{Lemma}

\theoremstyle{definition}
\newtheorem{defn}[thm]{Definition}
\newtheorem{exa}[thm]{Example}

\usepackage{scalerel}
\newlength\bshft
\def\fakebold#1{\ThisStyle{\ooalign{$\SavedStyle#1$\cr%
  \kern-\bshft$\SavedStyle#1$\cr%
  \kern\bshft$\SavedStyle#1$}}}

\newcommand{\smprod}{\textstyle\prod\limits}
\newcommand{\smsum}{\textstyle\sum\limits}
\newcommand{\defeq}{\mathrel{\mathop:}=}
\newcommand{\lto}{\longrightarrow}
\DeclareMathOperator{\quorem}{quo\_rem}
\DeclareMathOperator{\ev}{ev}
\DeclareMathOperator{\GR}{GR}
\newcommand{\mcB}{\mathcal B}
\newcommand{\mcC}{\mathcal C}
\newcommand{\mcM}{\mathcal M}
\newcommand{\mcN}{\mathcal N}
\newcommand{\mcT}{\mathcal T}

\newcommand{\ol}{\overline}
\newcommand{\F}{\mathbb F}
\newcommand{\Z}{\mathbb Z}
\let\theta\vartheta
\let\phi\varphi

\begin{document}

\title{List Decoding of Quaternary Codes\\ in the Lee Metric}

\author{\IEEEauthorblockN{Marcus Greferath}
  \IEEEauthorblockA{%
    \textit{School of Mathematics and Statistics} \\
    \textit{University College Dublin} \\
    Dublin, Republic of Ireland \\
    marcus.greferath@ucd.ie} \and
  \IEEEauthorblockN{Jens Zumbr\"agel}
  \IEEEauthorblockA{%
    \textit{Faculty of Computer Science and Mathematics} \\
    \textit{University of Passau} \\
    Passau, Germany \\
    jens.zumbraegel@uni-passau.de}}

\maketitle

\begin{abstract} 
  We present a list decoding algorithm for quaternary negacyclic codes
  over the Lee metric.  To achieve this result, we use a
  Sudan-Guruswami type list decoding algorithm for Reed-Solomon codes
  over certain ring alphabets.  Our decoding strategy for negacyclic
  codes over the ring~$\fakebold {\Z_4}$ combines the list decoding
  algorithm by Wu with the Gröbner basis approach for solving a key
  equation due to Byrne and Fitzpatrick.
\end{abstract}

\medskip



\section{Introduction}

It has been observed in the literature, that cyclic codes over the
alphabet~$\Z_4$ equipped with the Lee distance often have a larger
minimum distance and a better decoding capability than predicted by
their “designed distance”.  For this reason a list decoding approach
is suggested, although technical difficulties due to zero divisors in
rings are to be expected.

Codes over integer residue rings equipped with the Lee metric have
currently received increasing attention in the community. This stems
on one hand from the fact, that McEliece type cryptosystems based on
the Lee metric may offer increased security but lack so far the
availability of good codes.  On the other hand, connections between
lattice-based cryptography and coding theory are gaining attention, a
connection which is established by the Lee metric as it approximates
the Euclidean distance on lattices.

In this paper, we present a list decoding algorithm for quaternary
negacyclic codes with Lee distance.  To arrive there, we employ a
Sudan-Guruswami type list decoding algorithm for Reed-Solomon codes over
ring alphabets.  Note that this part is related to work by Armand~%
\cite{a05a, a05b} (see also~\cite{qbc13}), while our setup and
factorisation algorithm slightly differs.  For negacyclic codes over
the ring~$\Z_4$ (cf.~\cite{w99, bgpz13}), our decoding strategy
combines the list decoding algorithm by Wu~\cite{w08} with the Gröbner
basis approach for solving a key equation due to Byrne and
Fitzpatrick~\cite{bf02}.


\section{Preliminaries}

Let~$p$ be a prime and let~$m$ and~$r$ be positive integers.  We
denote by~$\F_{p^m}$ the finite field with~$p^m$ elements and let
$\GR(p^r, m)$ be the Galois ring of characteristic~$p^r$ and
degree~$m$.  The latter can be constructed as the quotient ring
$\Z_{p^r}[X] / (f)$ with monic polynomial $f \in \Z_{p^r}[X]$ of
degree~$m$ such that $f \bmod p$ in $\Z_p[X]$ is irreducible.  The
Galois ring $\GR(p^r, m)$ is a local ring with maximal ideal~$(p)$
such that all its ideals form a chain
$\{ (p^i) \mid 0 \le i \le r \}$, and one has the canonical
homomorphism onto the residue field
\[ \mu \colon \GR(p^r, m) \lto \F_{p^m} \,. \]

\subsection{Hensel lifting}

Hensel lifting is important for both the construction of Galois rings
and the factorisation of polynomials over such rings.  The setup
actually applies to a quite general situation,
following~\cite[Ch.~15]{vzgg}.  By a ring we mean a commutative ring
with identity.

Elements $g, h$ in a ring~$R$ are called \emph{Bézout-coprime} if
$s g + t h = 1$ for certain $s, t \in R$.  Note that for principal
ideal domains this definition amounts to the usual notion of
coprimeness of having no common factor.  But in general it is
stronger; in fact, coprime elements need not be Bézout-coprime,
consider e.g.\ $2, X \in \Z_4[X]$ or $X, Y \in \F_2[X, Y]$.

Now let~$R$ be a ring and $a \in R$.  The basic Hensel step allows to
lift a polynomial factorisation over the quotient ring $R / (a)$ to a
factorisation over the quotient ring $R / (a^2)$.  More precisely, 
suppose that we have $f^* \in (R / (a^2))[X]$ and $g, h \in
(R / (a))[X]$, with~$h$ monic, being Bézout-coprime such that
\[ f^* \bmod a \,=\, g \!\cdot\! h \,, \] then we find $g^*, h^* 
\in (R / (a^2))[X]$, with~$h^*$ monic, still Bézout-coprime such that
\[ f^* \,=\, g^* \!\cdot\! h^* \,. \]
Furthermore, given $s, t \in (R/(a))[X]$ such that $s g + t h = 1$ we
can also compute elements $s^*, t^* \in (R/(a^2))[X]$ satisfying
$s^* g^* + t^* h^* = 1$.  The details are given in Algorithm~%
\ref{alg:hensel} (cf.~\cite[Alg.~15.10]{vzgg}).

\begin{algorithm}
  \caption{Hensel step}\label{alg:hensel}
  \SetKwInOut{Input}{Input}
  \SetKwInOut{Output}{Output}
  \Input{$f^* \in (R/(a^2))[X]$ and $g, h, s, t \in (R/(a))[X]$,
    $h$~monic \\
    \hfill such that $f^* \bmod a = g h$ and $s g + t h = 1$}
  \Output{$g^*, h^*, s^*, t^* \in (R/(a^2))[X]$, $h^*$~monic \\
    \hfill such that $f^* = g^* h^*$ and $s^* g^* + t^* h^* = 1$\medskip} 
  coerce $g, h, s, t$ into $(R/(a^2))[X]$ \\[1mm]
  $e = f^* - g h$ \\
  $q, r = \quorem(s e, h)$ \hfill (note that $a \mid e, q, r$) \\
  $g^* = g + t e + q g$; $h^* = h + r$ \\[1mm]
  $b = s g^* + t h^* - 1$ \\
  $c, d = \quorem(s b, h^*)$ \hfill (note that $a \mid b, c, d$) \\
  $s^* = s - d$; \ $t^* = t - t b - c g^*$ \\[1mm]
  \Return{$g^*, h^*, s^*, t^*$}
\end{algorithm}

Applying the algorithm repeatedly, we can lift this way factorisations
modulo~$a$ to factorisations modulo $a^2, a^4, a^8$, etc.

\subsection{Bivariate polynomial factorisation}

Due to zero divisors in general rings one cannot expect their
polynomials to have as nice factorisation properties and algorithms as
in the field case.  Simple examples like
$X \!\cdot\! X = (X \!+\! 2) \!\cdot\! (X \!+\! 2) \in \Z_4[X]$
already show a non-unique factorisation behaviour.  However, provided
that the factors can be mapped into square-free Bézout-coprime
factors, say in the univariate polynomial ring over a field, we are
able to obtain a factorisation by Hensel lifting.

For the list decoding at hand we are interested in the factorisation
of a bivariate polynomial $Q \in R[X, Y]$ over a Galois ring
$R \defeq \GR(p^r, m)$.  Note that for Hensel lifting one cannot
simply use a factorisation as a bivariate polynomial over its residue
field $F \defeq \F_{p^m}$, as the prime factors in $F[X, Y]$ will
usually not be Bézout-coprime.

In the following we describe an adaption of the Zassenhaus
factorisation method (cf.~\cite[Sec.~15.6]{vzgg}) to work over Galois
rings.  The idea is to use Hensel lifting on two levels, first for
lifting a univariate polynomial factorisation in $F[X]$ to $R[X]$, and
then for lifting the factorisation in $(R[Y] / (Y \!-\! u))[X] \cong
R[X]$ to one in $(R[Y] / (Y \!-\! u)^{\ell})[X]$, from which we may
deduce the factorisation in $R[X, Y]$.  Details follow.

\begin{algorithm}\label{alg:factor}
  \caption{Bivariate polynomial factorisation}\label{alg:bivariate}
  \SetKwInOut{Input}{Input}
  \SetKwInOut{Output}{Output}
  \Input{$Q \in R[X, Y]$}
  \Output{factorisation $Q = Q_1 \!\cdot\! ... \!\cdot\! Q_t$
    into irreducibles \medskip}
  choose $u \in R$ such that $\mu Q(X, u) \in F[X]$ is square-free \\
  factorise $\mu Q(X, u) = d g_1 \!\cdot\! ... \!\cdot\! g_s$ over~$F$ \\
  use Hensel lifting to obtain $Q(X, u) = c f_1 \!\cdot\! ...
  \!\cdot\! f_s$ over~$R$ \\
  use Hensel lifting to obtain $\ol Q = C F_1 \!\cdot\! ...
  \!\cdot\! F_s$ over $R[Y] / (Y \!-\! u)^{\ell}$ \\
  combine factors to obtain $Q = Q_1 \!\cdot\! ... \!\cdot\! Q_t$
  over $R[Y]$
\end{algorithm}

\begin{enumerate}
\item Given $Q \in R[X, Y]$, we compute a univariate polynomial
  $Q_u \defeq Q(X, u) \in R[X]$ for some $u \in R$ and consider its
  reduction $\ol Q_u \defeq \mu(Q_u) \in F[X]$, which we can factorise
  by classical methods.  If the polynomial $\ol Q_u$ is square-free,
  so that there are no repeated factors, we proceed; otherwise we
  try a different choice of~$u$.
\item In the prime factorisation $\ol Q_u = d g_1 \!\cdot\! ...
  \!\cdot\! g_s$ over~$F$ (with $d \in F$ and the~$g_i$ monic),
  products of distinct factors are coprime and thus Bézout-coprime,
  since~$F[X]$ is a principal ideal domain.  Thus we may apply
  multifactor Hensel lifting (cf.~\cite[Alg.~15.17]{vzgg}) using the
  basic Hensel step (Algorithm~\ref{alg:hensel}) to obtain a
  factorisation of $Q_u = c f_1 \!\cdot\! ... \!\cdot\! f_s \in R[X]$
  into irreducibles, where $\mu f_i = g_i$ for all~$i$.
\item Given a factorisation of $Q_u = c f_1 \!\cdot\! ... \!\cdot\! f_s$
  in $R[X]$ into distinct Bézout-coprime irreducibles (with $c \in R$
  and the~$f_i$ monic), we apply Hensel lifting over the polynomial
  ring $R[Y]$ using the modulus $a \defeq Y \!-\! u$ to arrive at a
  factorisation $\ol Q = C F_1 \!\cdot\! ... \!\cdot\! F_s$ in
  $(R[Y] / (Y \!-\! u)^{\ell})[X]$ for some large enough~$\ell$.
\item We combine the factors $C, F_1, \dots, F_s$ into products $Q_1,
  \dots, Q_t$ (with $t \le s$) such that $Q = Q_1 \!\cdot\! ...
  \!\cdot\! Q_t$ holds in $R[X, Y]$, a step which is necessary as the
  quotient ring may introduce additional factors.  Since there is a
  bound in the $Y$-degree of the coefficients of the actual
  factors~$Q_i$, we can find them by computing products of
  $1, 2, 3, \dots$ factors until the bound is satisfied.
\end{enumerate}

We summarise our bivariate polynomial factorisation method in
Algorithm~\ref{alg:bivariate}.  Notice that all steps are polynomial
time except possibly for the last combine-factors step, which however
seems to be very efficient in practice.  We leave a more thorough
study of the factoring algorithm for future work.


\section{List decoding of Reed-Solomon codes\\
  over rings}\label{sec:rs}

List decoding of Reed-Solomon and related codes over Galois rings has
been considered by Armand~\cite{a05a, a05b}, while our setup and
factorisation algorithm is slightly different.  See also~\cite{qbc13}
for list decoding of Reed-Solomon codes over more general rings.  We
briefly present here the main concepts for Galois rings, as required
subsequently.

Let $R \defeq \GR(p^r, m)$ be a Galois ring of characteristic~$p^r$
and degree~$m$, and let $\theta \in R$ be an element of multiplicative
order~$p^m \!-\! 1$.  Such an element can be obtained by taking the
defining polynomial $f \in \Z_{p^r}[X]$ of~$R$ to be the Hensel lift
of a primitive polynomial over~$\Z_p$ of degree~$m$ and then letting~%
$\theta$ be the class of~$X$ modulo~$f$.

The set $\mcT \defeq \{ \theta^i \mid 0 \le i < p^m \!-\! 1 \}
\cup \{ 0 \} \subseteq R$ then maps bijectively onto the residue
field~$\F_{p^m}$ under the canonical map~$\mu$, and is called
\emph{Teichmüller set}.

\begin{defn}
  Given $n \le p^m$ and $1 \le k \le n$ as well as $\alpha_1, \dots,
  \alpha_n \in \mcT$ distinct, we define the $[n, k]$ \emph{Reed-%
    Solomon code} over~$R$ as the evaluation code 
  \[ \mcC \defeq \!\big\{ \ev(f) \!\defeq\! \big( f(\alpha_1), \dots,
    f(\alpha_n) \big) \mid f \in R[X] ,\, \deg f < k \big\} . \]
\end{defn}

\begin{lem}\label{lem:mindist}
  The (Hamming) minimum distance of~$\mcC$ equals $d \defeq n - k + 1$,
  thus the code is maximum distance separable.
\end{lem}

\begin{proof}
  Suppose that $c = \ev(f) \in C$ has weight $< d$, so that $f \in R[X]$
  has at least~$k$ zeros, say (w.l.o.g.)  $\alpha_1, \dots, \alpha_k$.
  Writing $f = \sum_{i=0}^{k-1} f_i X^i$ it follows that $(f_0, \dots,
  f_{k-1}) V = 0$ for the Vandermonde matrix $V \defeq ({\alpha_j}^i)_{ij}$,
  with \[ \det V = \smprod_{i < j} (\alpha_j \!-\! \alpha_i) \] a unit
  in~$R$, since $\mu \det V = \prod_{i < j} (\mu \alpha_j \!-\!
  \mu \alpha_i) \ne 0$.  Therefore, $f = 0$ and thus $c = 0$.
\end{proof}

We can in fact correct error weights beyond half the minimum distance
by adapting the list decoding approach by Sudan~\cite{s97}, as
described next.  It consists of an interpolation step which produces a
bivariate polynomial $Q \in R[X, Y]$, and a factorisation step using
Algorithm~\ref{alg:factor} by which the codewords within the list
decoding radius are obtained.

For the interpolation step we fix a finite set~$S$ of indices $(i, j)$
describing terms $X^i Y^j$, and require that~$S$ has more than~$n$
elements.  Given a received word $y \in R^n$ we consider the
interpolation problem
\[ Q(\alpha_i, y_i) = 0 \quad \text{for } 1 \le i \le n \,, \]
where $Q = \sum_{(i, j) \in S} c_{ij} X^i Y^j \in R[X, Y]$ with the
coefficients $c_{ij} \in R$ to be determined.  This is a linear system
with more variables than equations and thus contains a nonzero
solution.  Such a solution can be obtained using Smith normal form,
which is available over Galois rings as these are chain rings
(cf.~\cite[Sec.~2-D]{fnks14}).

Concretely, for a list error radius~$t$ we let \[ S \defeq \{ (i, j)
  \mid i + (k \!-\! 1) j \le n - t \} \,. \]
The next result shows that the interpolation polynomial~$Q$ carries
information on all codewords within distance~$\le t$.

\begin{lem}\label{lem:listdec}
  Suppose that $y = c + e$ with $c = \ev(f) \in \mcC$ for $f \in
  R[X]$, $\deg f < k$, and~$e$ an error vector of weight~$\le t$.
  Then $Q(X, f) = 0$.
\end{lem}

\begin{proof}
  Considering $h \defeq Q(X, f) = \sum_{(i, j) \in S} c_{ij} X^i f^j
  \in R[X]$, then since $\deg f < k$ and by the definition of~$S$,
  we see that $\deg h \le n - t$.  On the other hand, we have
  $y_i = c_i = f(\alpha_i)$ and thus $h(\alpha_i) = Q(\alpha_i,
  f(\alpha_i)) = Q(\alpha_i, y_i) = 0$ whenever $e_i = 0$, so for
  at least $n \!-\! t$ values~$\alpha_i$.  As in the proof of
  Lemma~\ref{lem:mindist} we can use the Vandermonde determinant
  and the fact that the $\mu \alpha_i$ are distinct to deduce
  that $h = 0$.
\end{proof}

\begin{exa}\label{exa:rs}
  Consider the $[64, 6]$ Reed-Solomon code over $R \defeq \GR(4, 6)$
  defined by the full Teichmüller set.  While the minimum distance is
  $d = 59$ by Lemma~\ref{lem:mindist} and thus the unique decoding
  radius is~$29$, we can list decode up to radius $t = 41$.  Indeed,
  the set $S \defeq \{ (i, j) \mid i \!+\! 5 j \le 23 \}$ is of
  cardinality~$65$, so we can conpute an interpolation polynomial~$Q$
  and in light of Lemma~\ref{lem:listdec} find the list of codewords
  by factorising this bivariate polynomial using
  Algorithm~\ref{alg:factor}.
\end{exa}

\subsection{Multiplicities}\label{ssec:guruswami}

We can also apply the Guruswami-Sudan list decoding approach~\cite{gs99}
incorporating multiplicities to the present situation.  For this we
alter the interpolation step such that every $(\alpha_i, y_i)$ should
be a zero $Q \in R[X, Y]$ with multiplicity~$e$, which means that for
$Q(X + \alpha_i, Y + y_i)$ every coefficient of $X^i Y^j$ with
$i + j < e$ vanishes.  This amounts to $\frac 1 2 e (e \!+\! 1)$
linear conditions for each point, so that we require the set~$S$ to
have more than $\frac 1 2 e (e \!+\! 1) n$ elements.  But now, if
there are~$t$ errors, the polynomials~$f$ for every codeword within
this radius the polynomial $h \defeq Q(X, f) \in R[X]$ has $n \!-\! t$
roots with multiplicity at least~$e$, i.e., $(X - \alpha)^e \mid h$.
This forces~$h$ to be zero, provided that we take
$S \defeq \{ (i, j) \mid i + (k \!-\! 1) j \le e (n \!-\! t) \}$, as
the next result shows.

\begin{lem}
  Let $h \in R[X]$ be a polynomial of degree $< k e$ such that
  $(X \!-\! \alpha_i)^e \mid h$ for at least~$k$ distinct $\alpha_i
  \in \mcT$.  Then~$h$ equals zero.
\end{lem}

\begin{proof}
  Over the residue field we have $(X \!-\! \mu \alpha_i)^e \mid \mu h$,
  where the $\mu \alpha_i \in F$ are distinct, hence we can argue by
  degrees to deduce $\mu h = 0$.  Therefore, $h = p \tilde h$ and
  we may view the polynomial $\tilde h$ over $\GR(p^{r-1}, m)$ with
  $\deg \tilde h = \deg h$ and still have $(X \!-\! \mu \alpha_i)^e
  \mid \mu \tilde h$ over the corresponding residue field.  Continuing
  this way, we see that $h = 0$.
\end{proof}

\begin{exa}
  Consider the $[64, 6]$ Reed-Solomon code over $R \defeq \GR(4, 6)$
  from Example~\ref{exa:rs}.  Using multiplicity $e = 2$, i.e., double
  roots, we can now list decode up to radius $t = 43$.  We take a set
  $S \defeq \{ (i, j) \mid i \!+\! 5 j \le 42 \}$ of cardinality
  $198 > 3 \cdot 64$, which guarantees the existence of an
  interpolation polynomial~$Q$.  Then every codeword polynomial~$f$
  within the decoding radius satisfies $Q(X, f) = 0$, so we can find
  these again by factoring~$Q$ with Algorithm~\ref{alg:factor}.
\end{exa}


\section{A jump into the Byrne-Fitzpatrick algorithm}%
\label{sec:bf}

The Byrne-Fitzpatrick algorithm~\cite{bf01, bf02} can be used to solve
key equations over Galois rings~$R$.  It considers for $U \in R[Z]$
the \emph{solution modules} \[ \mcM_k \defeq
  \big\{ (f, g) \in R[Z]^2 \mid U f \equiv g \bmod Z^k \big\} \,, \]
and given a Gröbner basis for~$\mcM_k$ refines it to a Gröbner basis
for~$\mcM_{k+1}$.  The method is dubbed “solution by approximations”
and is reminiscent of the Berlekamp-Massey algorithm.  Recently, the
algorithm was adapted to work also over skew polynomials over Galois
rings~\cite{prwz21}.

More specifically, one considers a term order on the set of terms
$(Z^j, 0)$ and $(0, Z^j)$, $j = 0, 1, 2, \dots$, so that leading term
and leading monomial of a nonzero pair $(f, g) \in R[Z]^2$ are
defined.  Given a module $\mcM \subseteq R[Z]^2$, a set
$\mcB \subseteq \mcM$ is called \emph{Gröbner basis} for~$\mcM$ if for
each $(f, g) \in \mcM$ there exists a Gröbner basis element such that
its leading monomial divides the leading monomial of $(f, g)$.

Now given a Gröbner basis~$\mcB_k$ for~$\mcM_k$, to construct a
Gröbner basis~$\mcB_{k+1}$ for~$\mcM_{k+1}$ one computes for each
$(f_i, g_i) \in \mcB_k$ the \emph{discrepancy}
\[ \zeta_i \defeq (U f_i - g_i)_k \in R \,, \]
where the subscript denotes the $k$-th coefficient.  Then if
$\zeta_i = 0$ we put $(f_i, g_i)$ into~$\mcB_{k+1}$.  Otherwise, we
look for some $(f_j, g_j) \in \mcB_k$ with smaller leading term such
that $\zeta_j \mid \zeta_i$, say $\zeta_i = q \zeta_j$ for some
$q \in R$, in which case we put $(f_i, g_i) - q (f_j, g_j)$ into~%
$\mcB_{k+1}$; if there is none, we put $Z (f_i, g_i)$ into~$\mcB_{k+1}$.

In order to adapt the list decoding approach by Wu~\cite{w08}, we are
interested in the following question: Given a Gröbner basis for~%
$\mcM_k$, how can we construct a Gröbner basis for~$\mcM_{k + \ell}$
with $\ell > 1$?  Thus we deal with a “jump” in the solution-by-%
approximations method.  At this point we are unable to fully solve this
problem, but we outline a method to construct elements in~%
$\mcM_{k + \ell}$ satisfying certain degree constraints, which will be
sufficient for the list decoding approach.

\begin{lem}\label{lem:ab}
  Given Gröbner basis elements $(f_i, g_i)$ and $(f_j, g_j)$
  in~$\mcB_k$ with leading coefficient a unit in~$R$, there exist
  polynomials $a, b \in R[Z]$ with unit leading coefficient and
  $\deg a + \deg b \le \ell$ such that $a (f_i, g_i) - b (f_j, g_j)
  \in \mcM_{k + \ell}$.
\end{lem}

\begin{proof}
  To address the problem at hand, we introduce for $(f_i, g_i) \in \mcB_k$
  the discrepancy polynomials \[ h_i \defeq
    \smsum_{\lambda = 0}^{\ell - 1} (U f_i - g_i)_{k + \lambda} \in R[Z] \,. \]
  Then we look for an expression $a h_i - b h_j = 0 \bmod Z^{\ell}$
  for some $a, b \in R[Z]$ of low degree, in which case we have
  $a (f_i, g_i) - b (f_j, g_j) \in \mcM_{k+\ell}$.  Such a pair
  $(a, b)$ can in turn be found by computing a Gröbner basis for the
  solution module \[ \mcN \defeq \big\{ (a, b) \in R[Z]^2
    \mid a h_i - b h_j \equiv 0 \bmod Z^{\ell} \big\} \]
  by a (slight adaption) of the Byrne-Fitzpatrick algorithm.
\end{proof}

In the decoding scenario described in the next section we do not know
the discrepancy polynomials and thus cannot compute the polynomials~$a$
and~$b$ directly, but we are able to deduce these by a list decoding
approach.


\section{List decoding of quaternary negacylic codes}

Let $n > 1$ be an odd integer.  By a quaternary \emph{negacyclic
  code} of length~$n$ we mean an ideal in the ring $\Z_4[X] / (X^n
- 1)$.  Such codes have been investigated by Wolfman~\cite{w99},
who examined their structure.  We equip the ring~$\Z_4$ with the
Lee weight $w(x) \defeq \min(x, 4 \!-\! x)$ and build upon the
algebraic decoding of Lee errors~\cite{bgpz13}.  Notice that the map
induced by $X \mapsto -X$ sends any negacyclic code isometrically
onto a cyclic code, though the negacyclic representation offers
some advantage regarding decoding.

As in the case of BCH codes we can specify negacyclic codes in terms
of roots.  For this we choose a Galois ring $R \defeq \GR(4, m)$ such
that $n \mid 2^m \!-\! 1$ together with an element~$\theta$ of order
$2^m \!-\! 1$ (see Section~\ref{sec:rs}).  Then there exists an
element~$\beta$ of order~$n$ and we fix a root $\alpha \defeq -\beta$
of order~$2 n$, satisfying $\alpha^n = -1$.

\begin{defn}
  The quaternary negacyclic code with~$t$ \emph{roots} $\alpha,
  \alpha^3, \dots, \alpha^{2 t - 1}$ is given by
  \[ \mcC \defeq \big\{ c \in \Z_4[X] / (X^n \!+\! 1) \mid
    c(\alpha^{2 i - 1}) = 0 \text{ for } 1 \le i \le t \big\} \,. \]
\end{defn}

It is shown~\cite[Thm.~1]{bgpz13} that the code~$\mcC$ has minimum
Hamming distance $\ge 2 t \!+\! 1$, so this clearly holds for the
minimum Lee distance, too.  An algebraic decoding method for errors up
to Lee weight~$t$ was devised~\cite{bgpz13}, based on a Gröbner
basis algorithm by Byrne and Fitzpatrick~\cite{bf02}.  However, it has
been observed that many such codes have a larger minimum Lee distance
than $2 t \!+\! 1$ (see Table~\ref{tab:lee}), which motiviates a list
decoding approach.

\begin{table}\centering
  \caption{Parameters of negacyclic codes of length~$n$, designed
    error-correcting capability~$t$, and rank~$k$ (i.e., size $4^k$).}
  \label{tab:lee}\medskip
  \begin{tabular}{ccccc}
    \hline\noalign{\smallskip}
    ~$n$~ & ~~$t$~~ & ~~$k$~~ & $2t\!+\!1$ & ~$d_{\rm Lee}$~ \\
    \noalign{\smallskip}
    \hline
    \noalign{\smallskip}
    15 & 1 & 11 & 3 & 3 \\
          & 2 & 7 & 5 & 5 \\
          & 3 & 5 & 7 & 10 \\
    \noalign{\smallskip}
    31 & 1 & 26 & 3 & 4 \\
          & 2 & 21 & 5 & 7 \\
          & 3 & 16 & 7 & 12 \\
          & 5 & 11 & 11 & 16 \\
          & 7 & 6 & 15 & 26 \\
    \hline
  \end{tabular}
\end{table}

\subsection{The key equation}

Here we adjust the list decoding method of Wu~\cite{w08} to the
present situation.  The central idea of this algorithm is to start
with a Berlekamp-Massey solution to the key equation, and to “refine”
it afterwards by formulating a list decoding problem.  We recall
therefore the key equation for negacyclic codes and outline its
algebraic decoding.

For an error vector $e \in \Z_4[X] / (X^n \!+\! 1)$ we define the
error locator polynomial \[ \sigma \defeq \smprod_{i=0}^{n-1}
  (1 \!-\! X_i Z)^{w(e_i)} \in R[Z] \,, \]
with $X_i \defeq \alpha^{-i}$ if $e_i = 1, 2$ and $X_i \defeq -
\alpha^{-i}$ if $e_i = 3$, so that $(1 \!-\! X_i Z)^{w(e_i)}$ equals
$1 \!-\! \alpha^i Z$ if $e_i = 1$, $(1 \pm \alpha^i Z)^2$ if $e_i = 2$
and $1 \!+\! \alpha^i Z$ if $e_i = 3$.  Then the error pattern is
completely determined by the roots~$\alpha^i$ for $0 \le i < 2 n$
of~$\sigma$.

We also let the syndrome polynomial be $s \defeq \sum_{i=1}^t
y(\alpha^{2 i - 1}) Z^{2 i - 1} = \sum_{i=1}^t e(\alpha^{2 i - 1})
Z^{2 i - 1} \in R[Z]$, which is known to the decoder.  This polynomial
determines an odd polynomial $u \defeq \sum_{i=1}^t u_{2 i - 1}
Z^{2 i - 1}$ by the equation $s (u^2 \!-\! 1) = Z u'$, which in turn
defines a polynomial $T \defeq \sum_{i=1}^t T_i Z^i$ by the relation
$(1 + T(Z^2)) (1 + Z u) \equiv 1 \pmod {Z^{2 t + 2}}$.  We arrive at
the key equation
\[ (1 + T) \, \phi \,\equiv\, \omega \pmod {Z^{t + 1}} \,, \]
from which we recover the even and the odd part of the error locator
polynomial~$\sigma$ by $\omega(Z^2) = \sigma_e$, $\phi(Z^2) =
\sigma_e + Z \sigma_o$.

This key equation can be solved by considering the solution module~%
$\mcM_{t + 1}$ as in Section~\ref{sec:bf} with $U \defeq 1 + T$ and
employing the Gröbner basis approach~\cite{bgpz13}.

\subsection{List decoding}

Even though we only know the syndrome polynomial up to degree
$2 t \!-\! 1$ and thus the polynomial~$T$ in the key equation up
to degree~$t$, we would like to correct more than~$t$ errors.
For this we pretend that we actually have access to more syndromes
and presume that we can set up a key equation modulo~$Z^{t + 1 + \ell}$.

Suppose that $(\phi_i, \omega_i)$ and $(\phi_j, \omega_j)$ have been
computed as Gröbner basis elements for~$\mcM_k$, then according to
Lemma~\ref{lem:ab} there is a solution $(\Phi, \Omega)$ for~%
$\mcM_{k + \ell}$ such that $a \phi_i - b \phi_j = \Phi$ and $a \omega_i
- b \omega_j = \Omega$.  Hence, for the error locator polynomial there
holds
\begin{align*} \Sigma
  &= \Sigma_e + \Sigma_o = \Omega(Z^2)
    + \tfrac 1 Z (\Phi(Z^2) \!-\! \Omega(Z^2)) \\
  &= a \big( \omega_i(Z^2) + \tfrac 1 Z (\phi_i(Z^2) 
    \!-\! \omega_i(Z^2) \big) \\
  &\qquad\quad - b \big( \omega_j(Z^2) + \tfrac 1 Z (\phi_j(Z^2) 
    \!-\! \omega_j(Z^2) \big) \\
  &= a \sigma_i - b \sigma_j \,.
\end{align*}
Therefore, whenever $\Sigma(\gamma) = 0$ then
\[ \frac {\sigma_j} {\sigma_i} (\gamma) = \frac a b (\gamma) \,, \]
which holds for the $\tau > t$ roots~$\gamma$ of~$\Sigma$.  This is a
rational approximation problem: We look for a rational function of
small degree that interpolates at least~$\tau$ out of~$2 n$ given
values.  Such kind of problem has been addressed by the list decoding
algorithm by Wu~\cite{w08}, which we adapt by our list decoding
algorithm of Section~\ref{sec:rs}.

More precisely, in the list decoding setup we have the~$2 n$
evaluation points $\alpha^i$ for $0 \le i < 2 n$, and we have~$\tau$
error positions~$\gamma$ for which
\[ (a \sigma_i \!-\! b \sigma_j) (\gamma) = 0 \,, \]
with $a, b \in R[Z]$ unknown of degree $\le \frac \ell 2$, where
$\ell \defeq \tau - t$.  In the context of Section~\ref{sec:rs} this
corresponds to an evaluation code of length~$2 n$, rank~$\ell \!+\! 1$
and with $2 n \!-\! \tau$ “errors”.  Taking into account the
particular form of the factors, a suitable set of indices is
\[ S \defeq \big\{ (i, j) \mid \max \{ i, \lfloor \tfrac \ell 2
  \rfloor j \} \le \tfrac \tau 2 \big\} \,, \]
which we require to have more than~$2 n$ elements in order to solve
the interpolation step (in the single-multiplicity $e = 1$ case).
One technical difficulty arising is that $\frac {\sigma_j} {\sigma_i}
(\gamma)$ might be infinity, in which case, following Wu~\cite{w08},
rather than $Q(\gamma, \infty) = 0$ we impose the linear condition
$\tilde Q(\gamma, 0) = 0$ with $\tilde Q \defeq Q(X, \frac 1 Y) Y^d$
the $Y$-reverse polynomial of~$Q$.

The bivariate polynomial $Q \in R[X, Y]$ has then the property
that the degree of the numerator of $Q(X, \frac a b)$ is at
most~$\tau$.  Suppose for now that the $\mu \gamma_j$ are distinct
for the~$\tau$ error positions $\gamma_i \in \langle \alpha \rangle$
(no “double error” occurs).  In that case we infer that $Q(X, \frac
a b) = 0$ by a similar proof as Lemma~\ref{lem:listdec}, and
hence we have a factor $b Y \!-\! a \mid Q$.

In the general case, following the strategy in~\cite[Sec.~7]{bgpz13},
we consider the reduction modulo the residue field~$F$ and have~$\tau$
error locations (possibly with multiplicity) $\mu \gamma$ such that
\[ \frac {\mu \sigma_j} {\mu \sigma_i} (\mu \gamma)
= \frac {\mu a} {\mu b} (\mu \gamma) \,. \]
Using the list decoding algorithm over fields by Wu~\cite{w08},
we can find~$\mu a$ and~$\mu b$.  Thus we compute $\mu \Sigma$,
by which we deduce all double errors $e_i = 2$ by its double
roots.  Then we subtract a vector consisting of only those
double roots, by which we are able to reduce the decoding problem
to the first case without double errors.

Observe that the cardinality of the set~$S$ equals
\[ (\lfloor \tfrac \tau 2 \rfloor \!+\! 1) ( \lfloor \tfrac \tau
  {2 \lfloor \ell / 2 \rfloor} \rfloor \!+\! 1) \,, \]
which exceeds~$2 n$ provided that $\tau^2 > 4 n \ell = 4 n (\tau
\!-\! t)$.  Incorporating sufficient large multiplicities~$e$, as in
Section~\ref{ssec:guruswami}, it suffices to require $\tau^2 > 2 n
(\tau \!-\! t)$, or $\tau < n \!-\! \sqrt {n (n \!-\! 2 t)}$.
Therefore, we arrive at the following result.

\begin{thm}
  For a quaternary negacyclic code of length~$n$ with~$t$ roots and
  designed distance $d = 2 t \!+\! 1$, the proposed list decoding
  algorithm corrects all codewords within radius~$\tau$ from the
  received word, provided that
  \[ \tau < n \!-\! \sqrt{n (n \!-\! d)} \,. \]
\end{thm}

\begin{exa}
  Let $n = 63$ and consider a quaternary negacyclic code with
  $t = 16$ roots and designed distance $2 t + 1 = 33$.  The
  algebraic decoding method~\cite{bgpz13} is thus able to correct
  up to~$16$ Lee weight errors.  With our list decoding approach
  (with multiplicity $e = 2$) we can correct however up to $\tau
  = 19$ Lee errors.  For this we let $S \defeq \{ (i, j) \mid i, j
  \le 19 \}$ of size $400 > 3 \!\cdot\! 2 n$, so we can construct a
  bivariate interpolation polynomial $Q \in R[X, Y]$ of max-degree~%
  $19$.  Provided that no double error occured, by factorising~$Q$
  using Algorithm~\ref{alg:factor} and looking for factors $b Y \!-\!
  a$ with $\deg a, \deg b \le 1$, we are able to solve the list
  decoding problem.  Otherwise, we employ the strategy outlined above.
\end{exa}
    


\begin{thebibliography}{cc}
\bibitem{a05a} M.\,A.~Armand,
  “List decoding of generalized Reed-Solomon codes over commutative rings,”
  IEEE Trans.\ Inf.\ Theory~51, no.~1 (2005), pp.~411--419
\bibitem{a05b} M.\,A.~Armand, “Improved list decoding
  of generalized Reed–Solomon and alternant codes over Galois rings,”
  IEEE Trans.\ Inf.\ Theory~51, no.~2 (2005), pp.~728--733
\bibitem{bf01} E.~Byrne and F.~Fitzpatrick,
  “Gröbner Bases over Galois Rings with an Application to Decoding
  Alternant Codes,”
  J.\ Symbolic Computation~31 (2001), pp.~565--584
\bibitem{bf02} E.~Byrne and F.~Fitzpatrick,
  “Hamming metric decoding of alternant codes over Galois rings,”
  IEEE Trans.\ Inf.\ Theory~48, no.~3 (2002), pp.~683--694
\bibitem{bgpz13} E.~Byrne, M.~Greferath, J.~Pernas, and J.~Zumbrägel,
  “Algebraic decoding of negacyclic codes over~$\Z_4$,”
  Des.~Codes Cryptogr.~66, no.~1 (2013), pp.~3--16
\bibitem{fnks14} C.~Feng, R.\,W.~Nóbrega, F.\,R.~Kschischang, and
  D.~Silva, “Communication over finite-chain-ring matrix channels,”
  IEEE Trans.\ Inf.\ Theory~60, no.~10 (2014), pp.~5899--5917
\bibitem{gs99} V.~Guruswami and M.~Sudan,
  “Improved decoding of Reed-Solomon and algebraic-geometry codes,”
  IEEE Trans.\ Inf.\ Theory~45, no.~6 (1999), pp.~1757--1767
\bibitem{vzgg} J.~von zur Gathen, J.~Gerhard,
  \emph{Modern Computer Algebra},
  Cambridge University Press, 2013
\bibitem{prwz21} S.~Puchinger, J.~Renner, A.~Wachter-Zeh, and J.~Zumbrägel,
  “Efficient Decoding of Gabidulin Codes over Galois Rings,”
  Proc.\ IEEE International Symposium on Information Theory (ISIT 2021),
  Melbourne, Australia
\bibitem{qbc13} G.~Quintin, M.~Barbier, C.~Chabot,
  “On generalized Reed–Solomon codes over commutative and noncommutative
  rings,” IEEE Trans.\ Inf.\ Theory~59, no.~9 (2013), pp.~5882--5897
\bibitem{s97} M.~Sudan,
  “Decoding of Reed Solomon codes beyond the error-correction bound,”
  J.\ Complexity~13, no.~1 (1997), pp.~180--193
\bibitem{w99} J.~Wolfman,
  “Negacyclic and cyclic codes over~$\Z_4$,”
  IEEE Trans.\ Inf.\ Theory~45, no.~7 (1999), pp.~2527--2532
\bibitem{w08} Y.~Wu,
  “New list decoding algorithms for Reed–Solomon and BCH codes,”
  IEEE Trans.\ Inf.\ Theory~54, no.~8 (2008), pp.~3611--3630
\end{thebibliography}
\end{document}